\documentclass[submission,copyright,creativecommons]{eptcs}
\providecommand{\event}{AMSLO23} 

\usepackage{iftex}
\usepackage{graphicx}
\usepackage{bussproofs}
\usepackage{tikz-cd}
\usepackage{hyperref}                         

\usepackage{doi} 

\hypersetup{
    colorlinks=true,
    linkcolor=blue,
    filecolor=magenta,      
    urlcolor=cyan,
}
\usepackage{amsmath,amssymb,url,color,amsfonts}
\usepackage{comment}
\usepackage{mathtools}
\usepackage{libertine}
\usepackage{amsthm}
\usepackage{lineno}
\usepackage{stmaryrd}
\usepackage{MnSymbol}
\usepackage{mathdesign}
\usepackage{cmll}
\usepackage{enumerate}

\newcommand{\GBL}{{\bf GBL}}
\newcommand{\GBLe}{{\bf GBL}_{ewf}}

\newcommand{\BL}{\bf BL}
\newcommand{\GD}{\bf GD}
\newcommand{\MV}{\bf MV}

\newcommand{\cinf}{\lfloor \mathbf{inf} \rfloor}
\newcommand{\Kmodels}{\Vdash^{\rm K}}

\newcommand{\GBMmodels}{\Vdash^{\rm GBM}}

\newcommand{\LBMmodels}{\Vdash^{\rm LBM}}
\newcommand{\nLBMmodels}{\nVdash^{\rm LBM}}

\newcommand{\LKmodels}{\Vdash^{\rm LK}}

\newtheorem{thm}{Theorem}
\newtheorem{lemma}[thm]{Lemma}
\newtheorem{proposition}[thm]{Proposition}%
\newtheorem{corollary}[thm]{Corollary}%
\newtheorem{note}[thm]{Note}

\theoremstyle{definition}
\newtheorem{definition}[thm]{Definition}%

\title{A Kripke Semantics for Hajek's BL}
\author{Andrew Lewis-Smith
\institute{Department of Computer Science\\
University of Sheffield,\\
Sheffield, United Kingdom}
\email{andrew.lewis-smith@sheffield.ac.uk}}
\def\titlerunning{A Kripke Semantics for Hajek's BL}
\def\authorrunning{Andrew Lewis-Smith}
\begin{document}
\maketitle

\begin{abstract}
We provide a generalisation of Kripke semantics for Petr Hajek's Basic Logic and prove soundness and completeness of the same with respect to our semantics. We find this semantics easily specialises to the linearly ordered Kripke frames for Godel-Dummett logic which $\BL$ properly contains. Our soundness, deduction theorem and completeness arguments further strengthen this analogy. This paper extends the insights of \cite{Robinson2021-LEWKSF} from $\GBLe$ to the case of $\BL$.
\end{abstract}

\section{Introduction}
\par Hajek's Basic logic ($\BL$) occupies a central place in contemporary research on fuzzy and substructural logic. $\BL$ is  primarily studied algebraically. This is only natural: the logic is strongly algebraizable \cite{Hajek1998-HAJMOF} and is the logic of t-norms \cite{Hajek1998}. But the papers of Jipsen and Montagna \cite{JIPSEN20101559} and Bova and Montagna \cite{BOVA20091143} suggest an alternative view of the situation. We can employ algebraic embedding results via poset products to construct generalisations of Kripke semantics appropriate to extensions of $\GBL$ \cite{Robinson2021-LEWKSF} (and later \cite{Fussner2022}). This situates systems such as $\GBL$ and $\BL$ among constructive and intermediate logics, whose relational semantics are well-understood.\\
\par The semantics we devise for $\BL$ restricts that of \cite{Robinson2021-LEWKSF}. The present structures are defined over linear frames, hence our designation `Linear Bova-Montagna structure' or LBM structure.\footnote{So-called on account of our study of Bova and Montagna's PSPACE-completeness results for the consequence relation of $\GBLe$-algebras in \cite{BOVA20091143}, which inspired our work in \cite{Robinson2021-LEWKSF}. We note in passing that the poset product construction itself appears to originate with Peter Jipsen (and is first recorded in Jipsen and Montagna \cite{JIPSEN20101559}), but as our previous publication refers to BM structures, and as LBMJ isn't as euphonious, we have determined to maintain precedent and refer to these as LBM structures.} We give a natural deduction system corresponding to the Hilbert system which we suspect in time can be adapted into a labelled calculus by importing insights from the semantics given here. For the present paper, these considerations yield a proof of soundness and completeness that attempts to de-emphasise reliance on algebra and resembles classic proofs of adequacy for systems like G{\"o}del-Dummett logic under Kripke semantics.\\
\par The structure of the paper is as with \cite{Robinson2021-LEWKSF}. Section 2 gives $\BL$'s natural deduction system, followed by suitable definitions of algebras, validity, and our relational semantics. We show how our semantics can specialise to the classic system of G{\"o}del-Dummett logic under linearly-ordered Kripke structures, and then prove $\BL$ sound and complete for our semantics.\\  

\section{Proof theory for Basic Logic}
\par We consider briefly the proof theory of $\BL$. We present the Hilbert-style and natural deduction renderings for the sake of clarity, but also to serve our later exposition and results (in particular our completeness proof).
\par The formulas of $\BL$ are inductively defined from atomic formulas, including $\bot$, and the binary connectives $\psi \land \chi$, $\psi \lor \chi$, $\psi \otimes \chi$ and $\psi \to \chi$. We will refer to this language as $\mathcal{L}_{\otimes}$, since it extends the language $\mathcal{L}$ of Godel-Dummett logic (see Note \ref{note of GD}) with a second form of conjunction, $\psi \otimes \chi$.
\par Figure \ref{table-bl} gives a natural deduction system for Hajek's $\BL$. When we write a sequent $\Gamma \vdash \phi$ we are always assuming $\Gamma$ to be a finite sequence of formulas. Note that we have the structural rules of weakening and exchange, but not contraction. Hence, the number of occurrences of a formula in $\Gamma$ matters, and one could think of the contexts $\Gamma$ as multisets. In particular, the rule $\to$ I removes one occurrence of $\phi$ from the context $\Gamma, \phi$, concluding $\phi \to \psi$ from the smaller context $\Gamma$. This makes $\BL$ a form of Affine logic. 
\begin{figure}
\begin{center}
\begin{tabular}{cc} 
\multicolumn{2}{c}{\AxiomC{}
    \RightLabel{\scriptsize{Ax}}
    \UnaryInfC{$\phi \vdash \phi$}
    \DisplayProof
    \qquad    
    \AxiomC{$\Gamma \vdash \psi$}
    \RightLabel{\scriptsize{W}}
    \UnaryInfC{$\Gamma, \phi \vdash \psi$}
    \DisplayProof
    \qquad
    \AxiomC{$\Gamma, \phi,\psi, \Delta \vdash \chi$}
    \RightLabel{\scriptsize{Ex}}
    \UnaryInfC{$\Gamma, \psi,\phi, \Delta \vdash \chi$}
    \DisplayProof
    }
 \\[4ex]
    \AxiomC{$\Gamma, \phi \vdash \psi$}
    \RightLabel{\scriptsize{$\to$ I}}
    \UnaryInfC{$\Gamma \vdash \phi \to \psi$}
    \DisplayProof
&   
    \AxiomC{$\Gamma \vdash \phi\to\psi$}
    \AxiomC{$\Delta \vdash \phi$}
    \RightLabel{\scriptsize{$\to$ E}}
    \BinaryInfC{$\Gamma, \Delta \vdash \psi$}
    \DisplayProof 
\\[4ex]
    \AxiomC{$\Gamma \vdash \phi$}
    \AxiomC{$\Delta \vdash \psi$}
    \RightLabel{\scriptsize{$\otimes$ I}}
    \BinaryInfC{$\Gamma, \Delta \vdash \phi \otimes \psi$}
    \DisplayProof
 &
    \AxiomC{$\Gamma, \phi, \psi \vdash \chi$}
    \AxiomC{$\Delta \vdash \phi \otimes \psi$}
    \RightLabel{\scriptsize{$\otimes$ E}}
    \BinaryInfC{$\Gamma, \Delta \vdash \chi$}
    \DisplayProof
 \\[4ex]
    \AxiomC{$\Gamma \vdash \phi$}
    \AxiomC{$\Gamma \vdash \psi$}
    \RightLabel{\scriptsize{$\land$ I}}
    \BinaryInfC{$\Gamma \vdash \phi \land \psi$}
    \DisplayProof
&
    \AxiomC{$\Gamma \vdash \phi_1 \land \phi_2$}
    \RightLabel{\scriptsize{$\land$ E}}
    \UnaryInfC{$\Gamma \vdash \phi_i$}
    \DisplayProof
 \\[4ex]
    \AxiomC{$\Gamma \vdash \phi_i$}
    \RightLabel{\scriptsize{$\lor$ I)}}
    \UnaryInfC{$\Gamma \vdash \phi_1 \lor \phi_2$}
    \DisplayProof
 &
    \AxiomC{$\Gamma \vdash \phi \lor \psi$}
    \AxiomC{$\Delta, \phi \vdash \chi$}
    \AxiomC{$\Delta, \psi \vdash \chi$}
    \RightLabel{\scriptsize{$\lor$ E}}
    \TrinaryInfC{$\Gamma, \Delta \vdash \chi$}
    \DisplayProof
\\[4ex]
    \AxiomC{$\Gamma, \phi, \phi \to \psi \vdash \chi$}
    \RightLabel{\scriptsize{DIV}}
    \UnaryInfC{$\Gamma, \psi, \psi \to \phi \vdash \chi$}
    \DisplayProof
&
    \AxiomC{$\Gamma \vdash \bot$}
    \RightLabel{\scriptsize{$\bot$ E}}
    \UnaryInfC{$\Gamma \vdash \phi$}
    \DisplayProof
\\[4ex]
    \AxiomC{}
    \RightLabel{\scriptsize{Prelin}}
    \UnaryInfC{$\Gamma \vdash \phi \to \psi \lor \psi \to \phi$}
    \DisplayProof
\end{tabular}
\caption{Basic logic  $\BL$}
\label{table-bl}
\end{center}
\end{figure}
{$\BL$} indeed has a \emph{resource sensitive} deduction theorem. The connective $\to$ internalises the consequence relation $\vdash$, and $\otimes$ internalises the comma in the sequent: 

\begin{proposition}\label{othersimple} 
The following hold in any calculus with rules Ax,$\to$~I, $\to$~E, $\otimes$~I, $\otimes$~E (and so for $\BL$:
\begin{enumerate}
\item $\Gamma, \psi \vdash \chi$ iff $\Gamma \vdash \psi \to \chi$. 
\item $\Gamma, \phi, \psi \vdash \chi$ iff $\Gamma, \phi\otimes\psi \vdash \chi$.
\end{enumerate}
\end{proposition}
\begin{proof}
As in \cite{Robinson2021-LEWKSF}.

%
    %
       
       
    %
%
%
\end{proof}

Below, we present the Hilbert system $\BL_{H}$ of (\cite{Hajek1998-HAJMOF}). 
\begin{itemize}
	\item[(A1)] $\phi \to \phi$
	\item[(A2)] $(\phi \to \psi) \to ((\psi \to \chi) \to (\phi \to \chi))$
	\item[(A3)] $(\phi \otimes \psi) \to (\psi \otimes \phi)$
	\item[(A4)] $(\phi \otimes \psi) \to \psi$
	\item[(A5)] $(\phi \to (\psi \to \chi)) \to ((\phi \otimes\psi) \to \chi))$
	\item[(A6)] $((\phi \otimes\psi) \to \chi)) \to (\phi \to (\psi \to \chi))$
	\item[(A7)] $(\phi \otimes (\phi \to \psi)) \to (\phi \wedge \psi)$
	\item[(A8)] $(\phi \wedge \psi) \to (\phi \otimes (\phi \to \psi))$
	\item[(A9)] $(\phi \wedge \psi) \to (\psi \wedge \phi)$
	\item[(A10)] $\phi \to (\phi \vee \psi)$
	\item[(A11)] $\psi \to (\phi \vee \psi)$
	\item[(A12)] $((\phi \to \psi) \wedge (\chi \to \psi)) \to ((\phi \vee \chi) \to \psi)$
	\item[(A13)] $\bot \to \phi$
	\item[(A14)] $(\phi \to \psi) \lor (\psi \to \phi)$
	\item[(R1)] $\phi, \phi \to \psi \vdash_{\scriptsize{\BL_{H}}} \psi$
\end{itemize}
thus obtaining the system we refer to as $\BL_{H}$. When we wish to stress the precise system in which a sequent $\Gamma \vdash \phi$ is derivable we use the system as a subscript of the provability sign, e.g. $\Gamma \vdash_{\mbox{\scriptsize{$\BL$}}} \phi$. 

\begin{proposition} \label{BL-equiv} The natural deduction system $\BL$ (Figure \ref{table-bl}) 
has the same derivable formulas as the Hilbert-style system $\BL_{H}$ of \cite{Hajek1998-HAJMOF}:
\[
\psi_1, \ldots, \psi_n \vdash_{\mbox{\scriptsize{$\BL$}}} \phi \quad \mbox{iff} \quad \vdash_{\scriptsize{\BL_{H}}} \psi_1 \to \ldots \to \psi_n \to \phi
\]
\end{proposition}
\begin{proof}
As in the proof of Proposition 1.5 in \cite{Robinson2021-LEWKSF}. Left-to-right follows by induction on the structure of the natural deduction proof once one shows each instance of a natural deduction rule translates to a theorem of $\BL_{H}$. We amend the provability ordering from \cite{Robinson2021-LEWKSF} for the present system: \[\phi \leq \psi\mbox{ iff }\vdash_{\scriptsize{\BL_{H}}} \phi\to\psi\]
Since $\BL_{H}$ results from $\GBLe$ by adding $(A14)$, all other cases are as in \cite{Robinson2021-LEWKSF}, except $(A14)$ which simply says that the provability relation is linearly ordered. For the right to left direction of the `iff', this follows by induction on the $\BL_{H}$ derivation of 
$\psi_1 \to \ldots \to \psi_n \to \phi$ once we show that each of the axioms of $\BL_{H}$ are theorems of {$\BL$}. The only case left to verify then is $(A14)$, and this is an axiom in the natural deduction calculus $\BL$ hence always provable in that calculus.
\end{proof}

\begin{note}\label{note of GD}
G{\"o}del-Dummett logic $\GD$ results from $\BL$ by adjoining the structural rule of contraction.
\end{note}

\subsection{BL and MV-algebras}
\label{Algebraic semantics for BL}
We situate the algebraic semantics characterising $\BL$ in terms of the somewhat larger theory of residuated lattices.


\begin{definition} $\mathcal{A} = \langle A, \land, \lor, \otimes, 1, \to \rangle$ is called a $\emph{commutative residuated lattice}$ if
\begin{itemize}
\item  $\langle A,\wedge,\vee,\otimes, 1 \rangle$ is a commutative lattice-ordered monoid
\item $x \otimes y \leq z$ if and only if $x \leq y \rightarrow z$
\end{itemize}
\end{definition}

\begin{definition}[$\BL$-algebras, $\BL$-chains]
A {\em $\BL$-algebra} is a bounded, commutative residuated lattice which satisfies the {\em divisibility property}: if $x\leq y$ then $y \otimes (y \rightarrow x) = x$; pre-linear: $(x \to y) \lor (y \to x)$; bounded from below by $\bot$, i.e. $\bot \leq x$ for all $x \in A$, and integral in that $1$ is the top element of the lattice, i.e. $x \leq 1$ for all $x \in A$. In this case we also denote $1$ by $\top$. Finally, we note the condition of the Divisibility property is equivalent to requiring that the residuated lattice satisfy the equation $x \otimes (x \rightarrow y) = y \otimes (y \rightarrow x)$. A $\BL$-chain is a totally-ordered $\BL$-algebra.
%
%
%
%
\end{definition}





\begin{definition}[$\MV$-algebra] A $\BL$-algebra is called an $\MV$-algebra if the negation map ($\neg x = x \rightarrow \bot$) is an involution, i.e. $(x\rightarrow \bot) \rightarrow \bot = x$, for all $x$.
\end{definition}

$\MV$-algebras provide an algebraic semantics for classical {\L}ukasiewicz logic. Here we are interested in a particular MV algebra which we will use in our Kripke semantics for $\BL$:

\begin{definition}[Standard $\MV$-chain] \label{std-mv-chain} For $x \in [0,1]$, let $\overline{x} := 1 - x$. The \emph{standard $\MV$-chain}, denoted $[0,1]_{{\rm MV}}$, is the $\MV$-algebra defined as follows: The domain of $[0,1]_{{\rm MV}}$ is the unit interval $[0,1]$, with the constants and binary operations defined as
\[
\begin{array}{lcl}
	\top & := & 1 \\[2mm]
	\bot & := & 0 \\[2mm]
	x \wedge y & := & \min \{ x, y \} \\[2mm]
	x \vee y & := & \max \{ x, y \} \\[2mm]
	x \otimes y & := & \max \{ 0, \overline{\overline{x} + \overline{y}} \} \\[2mm]
	x \to y & := & \min \{ 1, \overline{\overline{y} - \overline{x}} \}
\end{array}
\]
%
\end{definition}

\begin{note} $x \otimes y$ is equivalent to $\max\{ 0, x + y - 1\}$, and
$x \to y$ is equivalent to $\min \{ 1, y - x + 1 \}$.
%
\end{note}
\section{Valid Sequents in \BL}
\begin{definition}[Denotation functions]
\label{sec:denotation-functions} Given a $\MV$-chain $[0,1]_{\mathcal{\MV}}$, and a mapping from propositional variables to elements of $[0,1]_{\mathcal{\MV}}$:  $$p \mapsto  \llbracket p \rrbracket \in {[0,1]_{\mathcal{\MV}}}$$ We thus refer to the denotation of a variable $p$ as $\llbracket p \rrbracket_{\mathcal{\MV}}$. We can extend that mapping to all formulas in the language of L in a straightforward way:
\[
\begin{array}{lcl}
	\llbracket \phi \otimes \psi \rrbracket_{\mathcal{\MV}} & := &  \llbracket \phi \rrbracket_{\mathcal{\MV}} \otimes \llbracket \psi \rrbracket_{\mathcal{\MV}} \\[2mm]
	\llbracket \phi \land \psi \rrbracket_{\mathcal{\MV}} & := &  \llbracket \phi \rrbracket_{\mathcal{\MV}} \land \llbracket \psi \rrbracket_{\mathcal{\MV}} \\[2mm]
	\llbracket \phi \lor \psi \rrbracket_{\mathcal{\MV}} & := &  \llbracket \phi \rrbracket_{\mathcal{\MV}} \lor \llbracket \psi \rrbracket_{\mathcal{\MV}} \\[2mm]
	\llbracket \phi \to \psi \rrbracket_{\mathcal{\MV}} & := &  \llbracket \phi \rrbracket_{\mathcal{\MV}} \to \llbracket \psi \rrbracket_{\mathcal{\MV}}
\end{array}
\]

\end{definition}

\begin{definition}[Validity]
A sequent $\phi_1, . . . , \phi_n \vdash_{\BL} \psi$ is then said to 
be valid in $\BL$-algebras, if $\llbracket \phi_1 \rrbracket \otimes . . . \otimes \llbracket \phi_n \rrbracket \le \llbracket \psi \rrbracket$
holds in $\BL$-algebras. A sequent is said to be \emph{valid} if it is valid in \emph{all} $\BL$-algebras. We can write this: $\Gamma \models_{\BL} \phi$ In the case where $\phi$ is valid in all $\BL$-algebras, we write $\models_{\BL}\phi$. The valid sequents, in the sense above, are precisely the ones provable in Basic Logic \cite{Hajek1998-HAJMOF}:
\end{definition}

\begin{proposition}\label{ALGcompletenessBL} A sequent $\Gamma \vdash \psi$ is $\BL$-valid iff it is provable in $\BL$.
\end{proposition}

\section{Kripke Semantics for \BL}\label{sec-semantics}

\begin{note}
The Kripke semantics for $\BL$ that we propose is a restriction of our semantics introduced in \cite{Robinson2021-LEWKSF}. We first need to define a particular class of functions from the set of worlds $W$ to MV-chains.
\end{note} 

\begin{definition}[Sloping functions] Let ${\cal W} = \langle W, \succeq \rangle$ be a linear order and $[0,1]_{\MV}$ a $\BL$-algebra. A function $f \colon W \to [0,1]_{\MV}$ is said to be a \emph{sloping function for $\BL$}  (hereon  \emph{sloping function}, or  \emph{sloping}) if $f(w) > \bot$ implies $\forall v \succ w (f(v) = \top)$.
\end{definition}


\begin{lemma} \label{lemma-monotone function} If $f \colon W \to [0,1]_{\MV}$ and $g \colon W \to [0,1]_{\MV}$ are sloping, then the following functions are also sloping:

\[
\begin{array}{lcl}
    (f \wedge g)(w) & := & \min \{ f w, g w \} \\[2mm]
	(f \vee g)(w) & := & \max \{ f w, g w \} \\[2mm]
	(f \otimes g)(w) & := & \max \{ 0, \overline{\overline{f w} + \overline{g w}} \}	
\end{array}
\]
\end{lemma}

\begin{proof} As in \cite{Robinson2021-LEWKSF}. 
%
\end{proof}

\begin{definition} \label{def:inf} Let $\lfloor \cdot \rfloor$ be the usual ``floor" operation on the standard MV-chain $[0,1]_{{\rm MV}}$, corresponding to the case distinction
\[
\lfloor x \rfloor := \left\{
    \begin{array}{ll}
         \top & \mbox{if} \quad  x = \top \\[2mm]
         \bot & \mbox{if} \quad x < \top
    \end{array}
\right.
\]
which is known as the ``Monteiro-Baaz $\Delta$-operator''. Given a (not necessarily sloping) function $f \colon W \to [0,1]$ and a $w \in W$, let us write $\inf_{v \succeq w}$ for the following construction:
\[ \cinf_{v \succeq w} f(v) := \min \{ f(w) , \inf_{v \succ w} \lfloor f(v) \rfloor \}  \]
where $\inf_{v \succ w} \lfloor f(v) \rfloor$ is the infimum of the set $\{ \lfloor f(v) \rfloor  : v \succ w \} \subseteq [0,1]$.
\end{definition}

\begin{lemma} \label{inf-sloping} This definition of $\inf_{v \succeq w}$ can also be equivalently written as
\[
\cinf_{v \succeq w} f(v) := \left\{
    \begin{array}{ll}
         f(w) & \mbox{if} \quad \forall v \succ w (f(v) = \top) \\[2mm]
         \bot & \mbox{if} \quad \exists v \succ w (f(v) < \top)
    \end{array}
\right.
\]
and for any $f \colon W \to [0,1]$ the function $\lambda w . \cinf_{v \succeq w} f(v)$ is a sloping function.
\end{lemma}

\begin{proof} First let us show that this is an equivalent definition. Consider two cases: \\[1mm]
{\bf Case 1}. $\forall v \succ w (f(v) = \top)$. In this case $\inf_{v \succ w} \lfloor f(v) \rfloor = \top$ and hence
\[ \cinf_{v \succeq w} f(v) =  \min \{ f(w) , \top  \} = f(w) \] 
{\bf Case 2}. $\exists v \succ w (f(v) < \top)$. In this case $\inf_{v \succ w} \lfloor f(v) \rfloor = \bot$
\[ \inf_{v \succeq w} f(v) =  \min \{ f(w) , \bot  \} = \bot \] 
In order to see that $\lambda w . \cinf_{v \succeq w} f(v)$ is a sloping function, assume that for some $w$ we have $\cinf_{v \succeq w} f(v) > \bot$, and let $w' \succ w$. By definition we have that $\forall v \succ w (f(v) = \top)$, and hence $f(w') = \top$ and $\forall v \succ w' (f(v) = \top)$, which implies $\cinf_{v \succeq w'} f(v) = \top$.
\end{proof}

\begin{definition} Let $[0,1]_{\MV}$ be a $\MV$-algebra.  A \emph{Linear Bova-Montagna structure} for $[0,1]_{\MV}\label{def-gbm-structure}$ (or LBM-structure) is a pair $\mathcal{M}_{[0,1]_{\MV}} = \langle {\cal W}, \LBMmodels \rangle$ where ${\cal W} = \langle W, \succeq \rangle$ is a linear order, and $\LBMmodels$ is an infix operator (on worlds and propositional variables) taking values in $[0,1]_{\MV}$, i.e. $(w \LBMmodels p) \in [0,1]_{\MV}$, such that for any propositional variable $p$ the function $\lambda w . (w \LBMmodels p) \colon W \to [0,1]_{\MV}$ is a sloping function.
\end{definition}

\begin{definition}[LBM Kripke Semantics for $\mathcal{L}_\otimes$] \label{def:kripke-sem} Given a LBM-structure $$\mathcal{M}_{[0,1]_{\MV}} = \langle {\cal W}, \LBMmodels \rangle$$ the valuation function $w \LBMmodels p$ on propositional variables $p$ can be extended to all $\mathcal{L}_\otimes$-formulas as:
\[
\begin{array}{lcl}
    w \LBMmodels \top & := & \top \\[1mm]
	w \LBMmodels \bot & := & \bot \\[1mm]
	w \LBMmodels \phi \wedge \psi & := & (w \LBMmodels \phi) \wedge (w \LBMmodels \psi) \\[1mm]
	w \LBMmodels \phi \vee \psi & := & (w \LBMmodels \phi) \vee (w \LBMmodels \psi) \\[1mm]
	w \LBMmodels \phi \otimes \psi & := & (w \LBMmodels \phi) \otimes (w \LBMmodels \psi) \\[1mm]
	w \LBMmodels \phi \to \psi & := & \cinf_{v \succeq w}((v \LBMmodels \phi) \to (v \LBMmodels \psi))
\end{array}
\]
where the operations on the right-hand side are the operations on $[0,1]_{\BL}$.
\end{definition}


\begin{lemma}\label{monotonicity-lemma-logical} For any formula $\phi$ the function $\lambda w . (w \LBMmodels \phi) \colon W \to [0,1]_{\MV}$ is a sloping function.
\end{lemma}

\begin{proof} By induction on $\phi$. The cases for $\psi \vee \xi, \psi \wedge \xi$ and $\psi \otimes \xi$ follow directly from Definition \ref{def:kripke-sem}. The case for $\psi \to \xi$ follows from \ref{inf-sloping}.
\end{proof}


%

\begin{lemma}(The sloping functions are linearly ordered in LBM's.)\label{sloping-functions-are-linearly-ordered} Let $f,g: W \to [0,1]_{\MV}$ be sloping for $\BL$. Then:

$$\forall v \succeq w: (f(v) \ge g(v)) \lor \forall v \succeq w: (g(v) \ge f(v))$$
\end{lemma}
\begin{proof} We prove $\neg\forall v \succeq w: (f(v) \ge g(v)) \Rightarrow \forall v \succeq w: (g(v) \ge f(v))$ as this is classically equivalent to the above statement. So assume that $\neg\forall v \succeq w: (f(v) \ge g(v))$. Then $\exists v \succeq w: (f(v) < g(v))$. But then $g(v) > \bot$; and since $f,g$ are sloping, this means for any $v\prime \succ v$ we have $g(v\prime) = \top$ and so $g(v\prime) \ge f(v\prime)$. On the other hand, for any $v\prime \prec v$, $f(v\prime) = \bot$ as $f$ is sloping, and since this is the least element of the ordering, in particular we have $g(v\prime) \ge f(v\prime)$. In either case, we have $\forall v \succeq w: (g(v) \ge f(v))$ as desired.
\end{proof}

We can now generalise the monotonicity property of G{\"o}del-Dummett logic (under linearly-ordered Kripke frames) to $\BL$:

\begin{corollary}[Monotonicity] The following (generalised) monotonicity property holds for all $\mathcal{L}_\otimes$-formulas $\phi$, i.e.
\[ \mbox{if $w \preceq v$ then $(w \LBMmodels \phi) \leq (v \LBMmodels \phi)$} \]
\end{corollary}
\begin{proof}
This follows from the observation that the valuations are sloping functions, which are in turn monotone functions.
\end{proof} 
\section{Validity under LBM structures}

\begin{definition} \label{def:model} Let $\Gamma = \psi_1, \ldots, \psi_n$. Consider the following definitions:
\begin{itemize}
    \item We say that a sequent $\Gamma \vdash \phi$ \emph{holds} in a LBM-structure $\mathcal{M}$ (written $\Gamma \LBMmodels_{\mathcal{M}} \phi$) if for all $w \in W$ we have 
    $$(w \LBMmodels \psi_1 \otimes \ldots \otimes \psi_n) \leq (w \LBMmodels \phi)$$
    
    Otherwise (i.e. if $\Gamma \not \vdash \phi$), we say that the sequent \emph{fails} $\mathcal{M}$ (written $\Gamma \not\GBMmodels_{\mathcal{M}} \phi$) and this means:

    $$\exists w \in W: (w \LBMmodels \psi_1 \otimes \ldots \otimes \psi_n) > (w \LBMmodels \phi)$$
    
    \item A sequent $\Gamma \vdash \phi$ is said to be valid under the LBM Kripke semantics for $\mathcal{L}_\otimes$ (written $\Gamma \LBMmodels \phi$) if $\Gamma \LBMmodels_{\mathcal{M}} \phi$ for all LBM-structures $\mathcal{M}$.
\end{itemize}
\end{definition}


\section{LBMs and Linear Kripke structures}
\label{sec-gbm-gen-kripke}

\begin{note}
Linear Bova-Montagna structures generalise linear Kripke structures, i.e. Kripke structures where the frame has a linear ordering.\footnote{We do not provide the definition here, although this can be found in standard textbooks e.g. \cite{priest2008introduction}.} This is because Kripke structures merely require the valuations 
$(w \LBMmodels p) \in [0,1]_{\MV}$ are always in the finite set $\{0, 1\}$ or $\{\bot, \top\}$. These can then be identified with the Booleans. Therefore, any Linear Kripke structure can be seen as a LBM-structure, by defining
\end{note} 
\[ 
w \LBMmodels \phi = \left\{
    \begin{array}{ll}
        \top & \mbox{if $w \LKmodels \phi$} \\[2mm]
        \bot & \mbox{if $w \not\LKmodels \phi$} \\[2mm]
    \end{array}
    \right.
\]

\begin{note}
Recall that $\mathcal{L} \subset \mathcal{L}_\otimes$, so any $\mathcal{L}$-formula is also an $\mathcal{L}_\otimes$-formula. 
\end{note} 

\begin{thm} \label{thm-general} For any Linear Kripke structure $\mathcal{LK} = \langle \mathcal{W}, \Kmodels \rangle$ and $\mathcal{L}$-formula $\phi$, we have $\forall w$:
\[ w \LKmodels \phi \quad \mbox{iff} \quad (w \LBMmodels \phi) = \top \]
\end{thm}
\begin{proof} By induction on the complexity of the formula $\phi$. The base case follows by definition. \\[2mm]
%
%
{\bf Induction step}: We consider the important case. Suppose the result holds for all sub-formulas of $\phi$: \\[1mm]
%
%
{\bf $\to$ Case.} $\phi = \psi \to \chi$. We use the fact that when restricted to Linear Kripke structures, $(v \LBMmodels \psi) \in \{ \top, \bot \}$ and $(v \LBMmodels \chi) \in \{ \top, \bot \}$, and hence
    \begin{itemize}
        \item[(i)]\label{item:compositionality} $\forall v \succeq w (((v \LBMmodels \psi) = \top) \to ((v \LBMmodels \chi) = \top))  \Leftrightarrow \forall v \succeq w ((v \LBMmodels \psi) \to (v \LBMmodels \chi)) = \top$
        \item[(ii)]\label{item:inf} $\forall v \succeq w (((v \LBMmodels \psi) \to (v \LBMmodels \chi)) = \top) \Leftrightarrow {\bf inf}_{v \succeq w} ((v \LBMmodels \psi) \to (v \LBMmodels \chi)) = \top$, i.e. the ${\bf \cinf}_{v \succeq w}$ translates directly into a universally quantifed expression, i.e. it is (again) a standard $\inf_{v \succeq w}$ operation (on a set).
    \end{itemize} 
    Therefore:
    \[
        \begin{array}{lcl}
            w \LKmodels \psi \to \chi 
                & \equiv & \forall v \succeq w ((v \LKmodels \psi) \to (v \LKmodels \chi)) \\[1mm]
                & \stackrel{\textup{(IH)}}{\Leftrightarrow} & 
                   \forall v \succeq w ((v \LBMmodels \psi) = \top \to (v \LBMmodels \chi) = \top) \\[2mm]
                & \stackrel{(\ref{item:compositionality})}{\Leftrightarrow} & \forall v \succeq w ((v \LBMmodels \psi) \to (v \LBMmodels \chi) = \top) \\[2mm]
                & \stackrel{(\ref{item:inf})}{\Leftrightarrow} & 
                   {\bf \cinf}_{v \succeq w} ((v \LBMmodels \psi) \to (v \LBMmodels \chi)) = \top \\[2mm]
                & \equiv & (w \LBMmodels \psi \to \chi) = \top
        \end{array}
    \]
which concludes the proof.
\end{proof}
\par We also note the following:

\begin{proposition}
\label{BM-generalise-LBM}
BM structures of \cite{Robinson2021-LEWKSF} generalise LBM structures.


\end{proposition}

\begin{proof}
This follows from the fact that all linear orders are partial orders.
\end{proof}
\section{Soundness}

We now prove the soundness of the Kripke semantics for $\BL$.

\label{sec-soundness}

\begin{thm}[Soundness] If $\Gamma \vdash_{\mbox{\scriptsize{\BL}}} \phi$ then $\Gamma \LBMmodels \phi$.
\end{thm}
\begin{proof} By induction on the derivation of $\Gamma \vdash \phi$. Assume $\Gamma = \psi_1, \ldots, \psi_n$ and let $\otimes \Gamma := \psi_1 \otimes \ldots \otimes \psi_n$. Fix a LBM-structure $\mathcal{M} = \langle {\cal W}, \LBMmodels \rangle$ with $\mathcal{W} = \langle W, \succeq \rangle$, and let $w \in W$. We exhibit only one case, as the rest of the proof is analogous to that of \cite{Robinson2021-LEWKSF}. \\[1mm]

 (PRELIN) $\Gamma \vdash (\phi \to \psi) \lor (\psi \to \phi)$. By Definition \ref{def:model}, we need to show:
    \[
        \begin{array}{lcl} 
    	w \LBMmodels (\otimes \Gamma)  
		& \stackrel{\footnotesize{(\mbox{L.\ref{def:model}})}}{\leq} & (w \LBMmodels (\phi \to \psi) \lor (\psi \to \phi)) = \top 
	\end{array}
   \]   
 which is equivalent to:
     \[
        \begin{array}{lcl} 
    	w \LBMmodels (\otimes \Gamma)  
		& \stackrel{\footnotesize{(\mbox{L.\ref{std-mv-chain}})}}{\leq} & \max \{ (w \LBMmodels (\phi \to \psi)), (w\LBMmodels (\psi \to \phi))\} = \top 
	\end{array}
   \]   
 where the right of the inequality means: Either $(w \LBMmodels (\phi \to \psi)) = \top$ or $(w\LBMmodels (\psi \to \phi)) = \top$. Here we break into cases. \\
 
 {\bf Case 1.} $(w \LBMmodels (\phi \to \psi)) = \top$. We have:
    \[
        \begin{array}{lcl}
            (w \LBMmodels (\phi \to \psi)) = \top 
                & \equiv & \cinf_{v \succeq w}((v \LBMmodels \phi) \to (v \LBMmodels \psi)) = \top \\[1mm]
                & \stackrel{}{\Leftrightarrow} & 
                   {\forall v: v \succeq w}((v \LBMmodels \phi) \le (v \LBMmodels \psi)) \\[2mm]
        \end{array}
    \]

{\bf Case 2.} $(w\LBMmodels (\psi \to \phi)) = \top$. We have:
    \[
        \begin{array}{lcl}
  (w \LBMmodels (\psi \to \phi)) = \top 
                & \equiv & \cinf_{v \succeq w}((v \LBMmodels \psi) \to (v \LBMmodels \phi)) = \top \\[1mm]
                & \stackrel{}{\Leftrightarrow} & 
                   {\forall v: v \succeq w}((v \LBMmodels \psi) \le (v \LBMmodels \phi)) \\[2mm]
        \end{array}
    \]
 These latter cases show that we must then prove:
\begin{equation}\label{equation:disj} 
\forall v: v \succeq w((v \LBMmodels \phi) \le (v \LBMmodels \psi)) \lor \forall v: v \succeq w((v \LBMmodels \psi) \le (v \LBMmodels \phi))\end{equation}
But by \ref{monotonicity-lemma-logical}, $\lambda w . (v \LBMmodels \phi) \colon W \to [0,1]_{\MV}$ and $\lambda w . (v \LBMmodels \psi) \colon W \to [0,1]_{\MV}$  are sloping functions; and by \ref{sloping-functions-are-linearly-ordered} the sloping functions for LBM's are linearly ordered, so that indeed \ref{equation:disj} above holds.  
\end{proof}

\section{LBM's and Poset Products}
\label{sec:poset-sum} 
\par 
A \emph{poset product} (cf. \cite{BOVA20091143} and \cite{JIPSEN20101559}) is defined over a poset ${\cal W} = \langle W, \preceq \rangle$ as the algebra ${\bf A}_{\cal W}$ of signature $\mathcal{L}_\otimes$ whose elements are sloping functions $f \colon W \to [0,1]_{\MV}$ and operations are defined as below:\\
\[
\begin{array}{lcl}
	(\bot)(w) & := & \bot \\[2mm]
	(f_1 \wedge f_2)(w) & := & \min \{ f_1 w, f_2 w \} \\[2mm]
	(f_1 \vee f_2)(w) & := & \max \{ f_1 w, f_2 w \} \\[2mm]
	(f_1 \otimes f_2)(w) & := & \max \{ 0, \overline{\overline{f_1 w} + \overline{f_2 w}} \} \\[2mm]
	(f_1 \to f_2)(w) & := & \left
	\{
    \begin{array}{ll}
         f_1(w) \to f_2(w) & \mbox{if} \quad \forall v \succ w (f_1(v) \leq f_2(v)) \\[2mm]
         \bot & \mbox{otherwise.}
    \end{array}
    \right.
\end{array}
\]
Since $f_1$ and $f_2$ are sloping functions, we have that
\[ \forall v \succ w (f_1(v) \leq f_2(v)) \quad \Leftrightarrow \quad \forall v \succ w ((f_1(v) \to f_2(v)) = \top) \]
Therefore, this last clause of the definition can be simplified to
\[
\begin{array}{lcl}
	(f_1 \to f_2)(w) & := & \cinf_{v \succeq w} (f_1(v) \to f_2(v))
\end{array}
\]
\begin{definition}[Poset Product semantics for $\mathcal{L}_\otimes$] Let ${\cal W} = \langle W, \preceq \rangle$ be a fixed linearly ordered poset, and ${\bf A}_{\cal W}$ be the poset product described above. Given $h : Atoms \to {\bf A}_{\cal W}$ an assignment of atomic formulas to elements of ${\bf A}_{\cal W}$, any formula $\phi$ can be mapped to an element $\llbracket\phi\rrbracket_h \in {\bf A}_{\cal W}$ as follows:
\[
\begin{array}{lcl}
	\llbracket p \rrbracket_h & := & h(p) \quad (\mbox{for atomic formulas $p$}) \\[2mm]
	\llbracket\bot\rrbracket_h & := & \bot \\[2mm]
	\llbracket \phi \wedge \psi\rrbracket_h & := & \llbracket\phi\rrbracket_h \wedge \llbracket\psi\rrbracket_h \\[2mm]
	\llbracket\phi \vee \psi\rrbracket_h & := & \llbracket\phi\rrbracket_h \vee \llbracket\psi\rrbracket_h \\[2mm]
	\llbracket\phi \otimes \psi\rrbracket_h & := & \llbracket\phi\rrbracket_h \otimes \llbracket\psi\rrbracket_h \\[2mm]
	\llbracket\phi \to \psi\rrbracket_h & := & \llbracket\phi\rrbracket_h \to \llbracket\psi\rrbracket_h
\end{array}
\]
A formula $\phi$ is said to be \emph{valid in ${\bf A}_{\cal W}$ under $h$} if for every $w \in W$
\[ 
\llbracket\phi\rrbracket_{h}^{{\bf A}_{\cal W}}(w) = \top 
\]
(which is $1$ in $[0,1]_{{\rm MV}}$). A formula $\phi$ is said to be \emph{valid in ${\bf A}_{\cal W}$} if it is valid in ${\bf A}_{\cal W}$ under $h$ for any possible mapping $h \colon Atoms \to {\bf A}_{\cal W}$.
\end{definition}
The next proposition follows from a more general theorem stated as Theorem 2.2(1) in \cite{Busaniche2018PosetPA} for poset products where the indexing set is a forest, noting here that every chain is trivially a forest:
\begin{proposition}\label{posetprodBL}
Let ${\mathcal W} = \langle W, \preceq \rangle$ be a linearly-ordered poset and $\{{ A}_w : w\in W\}$ an indexed collection of BL-chains. Then:

$${\bf A}_{\cal W} = \prod_{w \in \langle W, \preceq \rangle} { A}_w$$ 

or the poset product of this collection is a linearly ordered $\BL$-algebra, i.e. a $\BL$-chain.
\end{proposition}

Since all MV-chains are BL-chains, we can specialise this latter:
\begin{corollary}\label{posetprodBLasMV}
Let ${\cal W} = \langle W, \preceq \rangle$ be a linearly-ordered poset and $\{{A}_w : w\in W\}$ an indexed collection of MV-chains. Then:

$${\bf A}_{\cal W} = \prod_{w\in \langle W, \preceq\rangle} { A}_w$$ 

or the poset product of this collection is a linearly ordered $\BL$-algebra, i.e. a $\BL$-chain.
\end{corollary}

\begin{note}
We conclude this section by observing that given a poset product ${\bf A}_{\cal W}$ over a linearly-ordered poset ${\cal W} = \langle W, \succeq \rangle$ and a mapping $h \colon Atoms \to {\bf A}_{\cal W}$ of atomic formulas to elements of ${\bf A}_{\cal W}$, we can obtain a LBM structure $\mathcal{M}^{{\bf A}_{\cal W}} = \langle {\cal W}, \LBMmodels_h \rangle$, by taking 
\[ (w \LBMmodels_h p) := h(p)(w) \]
recalling that $h(p) \colon W \to [0,1]_{\MV}$ is a sloping function.
\end{note}

\begin{proposition} \label{prop-poset-lbm} Let ${\bf A}_{\cal W}$ be the poset product over a linearly ordered poset ${\cal W} = \langle W, \succeq \rangle$, and $h \colon Atoms \to {\bf A}_{\cal W}$ be a fixed mapping of atomic formulas to elements of ${\cal W}$. Let $\mathcal{M}^{{\bf A}_{\cal W}}$ be the LBM-structure defined above. Then, for any formula $\phi$ 
\[ (w \LBMmodels_h \phi) = \llbracket{\phi}\rrbracket^{{\bf A}_{\cal W}}_h(w) \]
\end{proposition}
\begin{proof} By induction on the complexity of $\phi$. 
\end{proof}

So we can always transform an interpretation of $\mathcal{L}_\otimes$ formulas in the poset product ${\bf A}_{\cal W}$ into a general Kripke semantics (on the Kripke frame ${\cal W}$) for $\mathcal{L}_\otimes$ formulas.

\section{Completeness of LBM-semantics}
In the present section we prove completeness of $\BL$ for the semantics presented above. Our proof is a departure from our earlier paper \cite{Robinson2021-LEWKSF}. There, we simply embed our semantics into the poset products of Jipsen and Montagna, and let them do the rest of the work.
Here, we can actually give a slightly more detailed argument making use of facts about $\BL$-algebras, poset products and ordinal sums that are unique to this setting (see \cite{Busaniche2018PosetPA} for more on this issue). \\ 
\par Our approach is also different from that of Wesley Fussner\footnote{To comment in slightly more detail: We appeal solely to facts known from the literature on $\BL$-chains and exploit the conditions under which ordinal sums and poset products coincide. We also rely on a fixed natural deduction system and deduction theorem (given earlier), thus resembling classic proofs of completeness such as that for Godel-Dummett logic, furthering our claim that the semantics given really generalises the classic Kripke semantics for $\GD$. We hope this makes for a more digestible proof for a broader logical audience.}, who takes a more general approach based on poset products of GBL-algebras and considers on a case-by-case basis (potentially infinitely many) axiomatic extensions of the base system $\GBL$, alias $\BL$ sans pre-linearity, exchange, ex-falso quodlibet and commutativity. Our proof is an alternative to Fussner's, albeit for $\BL$ solely. Indeed, the conditions that make our proof go through are unique to $\BL$, and cannot be generalised to $\GBL$. 

\begin{thm}[Completeness] If $\Gamma \LBMmodels \phi$ then $\Gamma \vdash_{\mbox{\scriptsize{\BL}}} \phi$.
\end{thm}

We prove this theorem by way of the following lemma.

\begin{lemma}
If a formula fails in $\BL$, then it fails in a linear BM structure.
\end{lemma}
\begin{proof}
Let $\Gamma \equiv \psi_1, \ldots, \psi_n$. Suppose $\Gamma \vdash \phi$ fails in $\BL$.
By Propositions \ref{othersimple} and \ref{BL-equiv}, it follows that 
\[
\not\vdash_{\scriptsize{\BL_{H}}} \psi_1 \to \ldots \to \psi_n  \to \phi
\]
By the algebraic completeness result for $\BL$~algebras with respect to the Hilbert-style proof system $\BL_{H}$ (see \cite{Hajek1998}), it follows that for some $\BL$-algebra $\mathcal{G}$ and some mapping $h \colon Atom \to \mathcal{G}$ from propositional variables to elements of $\mathcal{G}$, we have
\[
\llbracket{\psi_1 \to \ldots \to \psi_n  \to \phi}\rrbracket_h^{\mathcal{G}} \neq \top
\]
By (\cite[Theorem 1]{Hajek1998}) we can take $\mathcal{G}$ to be a $\BL$-chain, and by 
Montagna's Theorem 3 of \cite{Montagna2005}, we can take $\mathcal{G}$ as the ordinal sum of finitely many copies of $[0,1]_{\MV}$; and by Busaniche's Lemma 2.3 of \cite{Busaniche2018PosetPA}, this particular ordinal sum is isomorphic to a poset product of finitely many copies of $[0,1]_{\MV}$ (which product will therefore also be linearly ordered and finite by Corollary \ref{posetprodBLasMV}).\\[1mm] 
%
Hence by Proposition \ref{prop-poset-lbm} there exists a finite linear order $\langle W, \preceq \rangle$ and a map $h\prime: Atoms \to [0,1]_{\MV}$ from atoms to elements of the poset product ${\bf A}_{\cal W}$ such that for some $w \in W$:  
\[ \llbracket{\psi_1 \to \ldots \to \psi_n  \to \phi}\rrbracket_{h'}^{{\bf A}_{\cal W}}(w) \neq \top \]
By Proposition \ref{prop-poset-lbm}, we have a LBM-structure $\mathcal{M}^{{\bf A}_{\cal W}}$ such that for some $w \in W$
\[
(w \LBMmodels_{h'} \psi_1 \to \ldots \to \psi_n  \to \phi) \neq \top
\]
and hence
\[
(w \LBMmodels_{h'} \psi_1 \otimes \ldots \otimes \psi_n) \not\leq (w \LBMmodels \phi)
\]
i.e. $\psi_1, \ldots, \psi_n \nLBMmodels \phi$, so the sequent fails in a LBM.
\end{proof}

\section{Conclusion}
 The state of the art abounds with work deeply algebraic in character. By introducing generalisations of Kripke semantics adequate for $\BL$ (and neighbours) we seek a new  perspective on fuzzy logics such as $\BL$ as constructive or semi-constructive systems. This analogy is justified by a relational semantics that specialises in the present case to $\GD$. Moreover, relational semantics typically come pregnant with connections to proof theory, decidability and model theory. We hope in the fullness of time the semantics developed here (and in \cite{Robinson2021-LEWKSF} and \cite{Fussner2022}) will suggest new calculi for $\BL$ which currently lacks a suitable analytic, syntactically-based proof theory. The existent analytic calculi for $\BL$ rely crucially on semantic insights, e.g. \cite{phdthesis}.
\par There are three main alternatives on offer. One approach is semantic tableaux, or refutation systems. One can exploit the notion of unsatisfiability in our relational semantics to devise a proof system, as is done in classical modal or Intuitionistic logic (see \cite{priest2008introduction}). Another possibility is to consider labelled calculi, e.g. \cite{negri2011} or \cite{Gabbay1996-GABLDS}. Labelled calculi have been well-tested in cases of classical modal logics, extensions of classical modal systems, substructural cases, as well as logics over an intuitionistic or minimal base, although as yet untested on cases involving of divisibility ((A7)-(A8) p.4). Still another alternative are the multi-type calculi of e.g. \cite{DBLP:journals/sLogica/GrecoLMP21}. These calculi seem well-adapted to systems defined over distributive structures, e.g. chains, having strong algebraic underpinnings. This bodes well for $\BL$, which is complete for $\BL$-chains.

\bibliographystyle{eptcs}
\bibliography{generic}

\begin{thebibliography}{10}
\providecommand{\bibitemdeclare}[2]{}
\providecommand{\surnamestart}{}
\providecommand{\surnameend}{}
\providecommand{\urlprefix}{Available at }
\providecommand{\url}[1]{\texttt{#1}}
\providecommand{\href}[2]{\texttt{#2}}
\providecommand{\urlalt}[2]{\href{#1}{#2}}
\providecommand{\doi}[1]{doi:\urlalt{https://doi.org/#1}{#1}}
\providecommand{\eprint}[1]{arXiv:\urlalt{https://arxiv.org/abs/#1}{#1}}
\providecommand{\bibinfo}[2]{#2}

\bibitemdeclare{article}{BOVA20091143}
\bibitem{BOVA20091143}
\bibinfo{author}{Simone \surnamestart Bova\surnameend} \&
  \bibinfo{author}{Franco \surnamestart Montagna\surnameend}
  (\bibinfo{year}{2009}): \emph{\bibinfo{title}{The consequence relation in the
  logic of commutative {GBL}-algebras is {PSPACE}-complete}}.
\newblock {\slshape \bibinfo{journal}{Theoretical Computer Science}}
  \bibinfo{volume}{410}(\bibinfo{number}{12}):\bibinfo{eid}{4},
  \doi{10.1016/j.tcs.2008.10.024}.
\newblock \eprint{2042.12345}.

\bibitemdeclare{article}{Busaniche2018PosetPA}
\bibitem{Busaniche2018PosetPA}
\bibinfo{author}{M.~\surnamestart Busaniche\surnameend} \&
  \bibinfo{author}{C.~\surnamestart Gomez\surnameend} (\bibinfo{year}{2018}):
  \emph{\bibinfo{title}{Poset Product and BL-Chains}}.
\newblock {\slshape \bibinfo{journal}{Studia Logica}} \bibinfo{volume}{106},
  \doi{10.1007/s11225-017-9764-6}.

\bibitemdeclare{article}{Fussner2022}
\bibitem{Fussner2022}
\bibinfo{author}{Wesley \surnamestart Fussner\surnameend}
  (\bibinfo{year}{2021}): \emph{\bibinfo{title}{Poset Products as Relational
  Models}}.
\newblock {\slshape \bibinfo{journal}{Studia Logica}}
  \bibinfo{volume}{110}:\bibinfo{eid}{4}, \doi{10.1007/s11225-021-09956-z}.

\bibitemdeclare{inbook}{Gabbay1996-GABLDS}
\bibitem{Gabbay1996-GABLDS}
\bibinfo{author}{Dov \surnamestart Gabbay\surnameend} (\bibinfo{year}{1996}):
  \emph{\bibinfo{title}{Labelled Deductive Systems}}.
\newblock \bibinfo{publisher}{Oxford University Press}.

\bibitemdeclare{article}{DBLP:journals/sLogica/GrecoLMP21}
\bibitem{DBLP:journals/sLogica/GrecoLMP21}
\bibinfo{author}{Giuseppe \surnamestart Greco\surnameend}, \bibinfo{author}{Fei
  \surnamestart Liang\surnameend}, \bibinfo{author}{Michael~Andrew
  \surnamestart Moshier\surnameend} \& \bibinfo{author}{Alessandra
  \surnamestart Palmigiano\surnameend} (\bibinfo{year}{2021}):
  \emph{\bibinfo{title}{Semi De Morgan Logic Properly Displayed}}.
\newblock {\slshape \bibinfo{journal}{Studia Logica}} \bibinfo{volume}{109},
  \doi{10.1007/s005000000044}.

\bibitemdeclare{article}{Hajek1998}
\bibitem{Hajek1998}
\bibinfo{author}{Petr \surnamestart Hajek\surnameend} (\bibinfo{year}{1998}):
  \emph{\bibinfo{title}{Basic Fuzzy Logic and BL Algebras}}.
\newblock {\slshape \bibinfo{journal}{Soft Computing}} \bibinfo{volume}{2},
  \doi{10.1007/s005000050043}.

\bibitemdeclare{inbook}{Hajek1998-HAJMOF}
\bibitem{Hajek1998-HAJMOF}
\bibinfo{author}{Petr \surnamestart Hajek\surnameend} (\bibinfo{year}{1998}):
  \emph{\bibinfo{title}{Metamathematics of Fuzzy Logic}}.
\newblock {\slshape \bibinfo{series}{Trends in Logic}}~\bibinfo{volume}{4},
  \bibinfo{publisher}{Kluwer}, \doi{10.1007/978-94-011-5300-3_4}.

\bibitemdeclare{article}{JIPSEN20101559}
\bibitem{JIPSEN20101559}
\bibinfo{author}{Peter \surnamestart Jipsen\surnameend} \&
  \bibinfo{author}{Franco \surnamestart Montagna\surnameend}
  (\bibinfo{year}{2010}): \emph{\bibinfo{title}{Embedding theorems for classes
  of {GBL}-algebras}}.
\newblock {\slshape \bibinfo{journal}{Journal of Pure and Applied Algebra}}
  \bibinfo{volume}{214}(\bibinfo{number}{9}):\bibinfo{eid}{4},
  \doi{10.1016/j.jpaa.2009.11.015}.

\bibitemdeclare{phdthesis}{phdthesis}
\bibitem{phdthesis}
\bibinfo{author}{Agnieszka \surnamestart Kulacka\surnameend}
  (\bibinfo{year}{2017}): \emph{\bibinfo{title}{Propositional Fuzzy Logics:
  Tableaux and Strong Completeness}}.
\newblock Ph.D. thesis, \bibinfo{school}{Imperial College London},
  \doi{10.25560/68598}.
\newblock \urlprefix\url{https://spiral.imperial.ac.uk/handle/10044/1/68598}.

\bibitemdeclare{article}{Robinson2021-LEWKSF}
\bibitem{Robinson2021-LEWKSF}
\bibinfo{author}{Andrew \surnamestart Lewis-Smith\surnameend},
  \bibinfo{author}{Paolo \surnamestart Oliva\surnameend} \&
  \bibinfo{author}{Edmund \surnamestart Robinson\surnameend}
  (\bibinfo{year}{2020}): \emph{\bibinfo{title}{A Kripke Semantics for
  Intuitionistic \L{}ukasiewicz Logic}}.
\newblock {\slshape \bibinfo{journal}{Studia Logica}}
  \bibinfo{volume}{109}:\bibinfo{eid}{4}, \doi{10.1007/s11225-020-09908-z}.

\bibitemdeclare{article}{Montagna2005}
\bibitem{Montagna2005}
\bibinfo{author}{Franco \surnamestart Montagna\surnameend}
  (\bibinfo{year}{2005}): \emph{\bibinfo{title}{Generating the Variety of
  BL-Algebras}}.
\newblock {\slshape \bibinfo{journal}{Studia Logica}} \bibinfo{volume}{9},
  \doi{10.1007/s00500-004-0450-z}.

\bibitemdeclare{inbook}{negri2011}
\bibitem{negri2011}
\bibinfo{author}{Sara \surnamestart Negri\surnameend} \&
  \bibinfo{author}{Jan~Van \surnamestart Plato\surnameend}
  (\bibinfo{year}{2011}): \emph{\bibinfo{title}{Proof Analysis: A Contribution
  to Hilbert's Last Problem}}.
\newblock \bibinfo{publisher}{Cambridge University Press},
  \doi{10.1017/CBO9781139003513}.

\bibitemdeclare{inbook}{priest2008introduction}
\bibitem{priest2008introduction}
\bibinfo{author}{Graham \surnamestart Priest\surnameend}
  (\bibinfo{year}{2008}): \emph{\bibinfo{title}{An Introduction to
  Non-Classical Logic: From If to Is}}.
\newblock \bibinfo{publisher}{Cambridge University Press},
  \doi{10.1017/CBO9780511801174}.

\end{thebibliography}
\end{document}